\documentclass[11pt]{article}
\usepackage[letterpaper]{geometry}

\usepackage[utf8]{inputenc}
\usepackage[OT4]{fontenc}
\usepackage{amsmath}
\usepackage{amsthm}
\usepackage{amssymb}
\usepackage{tikz}
\usepackage[retainorgcmds]{IEEEtrantools}
\usepackage{thmtools}
\usepackage{thm-restate}
\usepackage{url}
\usepackage{hyperref}
\usepackage[retainorgcmds]{IEEEtrantools}
\usepackage{thm-restate}

\newtheorem{definition}{Definition}[section]
\newtheorem{theorem}[definition]{Theorem}
\newtheorem{corollary}[definition]{Corollary}

\newtheorem{lemma}[definition]{Lemma}
\newtheorem{claim}[definition]{Claim}

\newcommand{\Pd}{\mathsf{P}}
\DeclareMathOperator*{\EE}{E}

\def\noqed{\renewcommand{\qedsymbol}{}}

\title{Upper Tail Estimates with Combinatorial Proofs}

\author{Jan Hązła\thanks{ETH Z\"urich, Department of Computer Science, Zurich, Switzerland. E-mail: 
{\tt \{jan.hazla,thomas.holenstein\}@inf.ethz.ch}}
\and 
Thomas Holenstein\footnotemark[1]}

\begin{document}

\maketitle

\begin{abstract}
We study generalisations of a simple, combinatorial
proof of a Chernoff bound
similar to the one by Impagliazzo and Kabanets (RANDOM, 2010).

In particular, we prove a randomized version
of the hitting property of expander random
walks and use it to obtain
an optimal expander random
walk concentration bound settling a question 
asked by Impagliazzo and Kabanets.  

Next, we obtain an upper tail bound for 
polynomials with input variables in $[0, 1]$ which are
not necessarily independent, but obey a certain condition
inspired by Impagliazzo and Kabanets. The resulting bound 
is applied by Holenstein and Sinha (FOCS, 2012)
in the proof of a lower bound for the number of calls
in a black-box construction of a pseudorandom
generator from a one-way function.

We also show that the same technique yields
the upper tail bound for the number of copies
of a fixed graph in an Erdős–Rényi random graph,
matching the one given by Janson, Oleszkiewicz, and 
Ruciński (Israel J.~Math, 2002).
\end{abstract}

\section{Introduction}

\paragraph*{Motivation and previous work}
Concentration bounds are inequalities that estimate the probability
of a random variable assuming a value that is far from its expectation.
They have a multitude of applications all across
the mathematics and theoretical computer science. See, e.g., textbooks
\cite{MR95, MU05, AB09, DP09}
for uses in complexity theory and randomised algorithms.

A typical setting is when this variable is a function 
$f(x)$ of $n$ simpler random variables $x = (x_1,\ldots,x_n)$ 
that possess a certain degree of independence
and we try to bound said probability with a function
decaying exponentially with $n$ (or, maybe, $n^{\epsilon}$ for some $\epsilon > 0$).

The canonical examples are Chernoff-Hoeffding bounds
\cite{Che52, Hoe63} for the sum of $n$ independent 
random variables in $[0,1]$ and 
Azuma's inequality \cite{Azu67}
for martingales.

The standard technique to prove Chernoff bounds 
is due to Bernstein \cite{Ber24}.
The idea is to bound $\EE[e^{tf(x)}]$ for
some appropriately chosen $t$, and then to apply
Markov's inequality.

Recently, Impagliazzo and Kabanets \cite{IK10}
gave a different, combinatorial proof of Chernoff bound,
arguing that its simplicity and nature provide 
additional insight into understanding concentration. 
What is more, their proof
is constructive in a certain sense (see \cite{IK10} for details).

The proof given by Impagliazzo and Kabanets is related to previous
published results: 
in \cite{SSS95}, Schmidt, Siegel and Srinivasan
give a Chernoff bound which is 
applicable in case the random variables $x=(x_1,\ldots,x_n)$ 
are only $m$-wise independent for some large enough $m$.
It turns out that the expressions which appear in their computations
have close counterparts in the proof in \cite{IK10},
but they still bound $\EE[e^{tf(x)}]$, and it seems to us
that the approach in \cite{IK10} makes the concepts clearer and the
calculations shorter.

Another work related to \cite{IK10} is due to
Janson, Oleszkiewicz and Ruciński \cite{JOR02}, who
give an upper tail bound (i.e., a one-sided
concentration bound) for the number of subgraphs in 
an Erd\H{o}s-R\'enyi random graph $\mathsf{G}_{n,p}$.
The proof given in \cite{JOR02} bears much relationship
to the proof given in \cite{IK10}.
We elaborate on that in Section~\ref{sec:ik-vs-jor}.

Finally, there is a connection to an argument used
by Rao to prove a concentration
bound for parallel repetition of two-prover games \cite{Rao08}. 
As we will see, one of the ideas in the proof given in \cite{IK10}
is to consider a subset of the variables $(x_1,\ldots,x_n)$.
Rao also does this, with a somewhat different purpose.

\paragraph*{Our contributions}
In this paper we modify the proof of Impagliazzo and Kabanets
and introduce a more general sufficient condition for concentration
which we term \emph{growth boundedness} (Section \ref{sec:chernoff}). 
Then, we show some applications of our framework.

First, we prove a randomized version of the hitting property
of expander random walks (Theorem \ref{thm:expander-hitting}) and use it to 
obtain an optimal (up to a constant factor in
the exponent) expander random walk concentration bound
settling a question asked in \cite{IK10} (Theorem \ref{thm:expander-main}).\footnote{
	Of course the bound itself is not new. Impagliazzo and Kabanets
	asked if such a concentration bound can be obtained
	from the hitting property, i.e., using the technique
	from \cite{IK10}.
}
We also show that our method is quite robust:
with a little more effort one
can improve the constant factor to the optimal
one in case of large number of steps and small
deviation (Theorem \ref{thm:expander-tight}).

Second, we prove an
upper tail bound for polynomials with input
random variables in $[0,1]$ (Theorem \ref{thm:kv-simple}).
Contrary to the previous work we are aware of,
we do not assume that those variables
are independent, but rather that they obey
a condition similar to growth boundedness.

This bound is used in a proof of a lower bound
for the complexity of a black-box construction
of a pseudorandom generator from a one-way
function \cite{HS12}. Although \cite{HS12}
was published earlier,
the proof of the bound is not contained there,
but deferred to this paper instead.
We outline how the bound was used
in \cite{HS12} in Section \ref{sec:hs12}.

\paragraph*{Notation}
Throughout the paper we focus on the bounds
of the form
 $\Pr[f(x) \ge \mu(1+\epsilon)]$).
We call such bounds ``(multiplicative) upper tail bounds''.

Typically, we consider a probability distribution $\Pd_x$ 
over some vector of random variables
$x = (x_1, \ldots, x_n)$.
We denote a random choice from $\Pd_x$ as $x \leftarrow \Pd_x$.
We try to explicitly indicate randomness whenever taking probability
or expectation, i.e., we write $\Pr_{x \leftarrow \Pd_x}\left[\ldots\right]$
and so on. For a finite set $A$, let $a \leftarrow A$ be a shorthand for a uniform random choice 
of an element from $A$.

For a natural number $n$, let $[n] := \{1, \ldots, n\}.$ 
As usual, by $\binom{n}{k}$ we denote 
$\frac{\prod_{i=0}^{k-1} (n-i)}{k!}$ for $n \in \mathbb{R}$ and
$k \in \mathbb{N}$. For $n \in \mathbb{N}$ and $0 \le k \le n$, we also identify
$\binom{n}{k}$ with the set of subsets of $[n]$ of size $k$.

In particular, $(i_1, \ldots, i_m) \leftarrow [n]^m$ denotes
uniform choice of $m$ elements from $[n]$ with repetition
and $M \leftarrow \binom{n}{m}$ uniform choice
of a subset of $[n]$ of size $m$.

\section{A Simple Proof of a Chernoff Bound}

We start by presenting a short proof of
a Chernoff bound in, arguably, the most basic setting. 
\begin{theorem}\label{thm:toy-chernoff}
Let $x = (x_1, \ldots, x_n)$ be i.i.d.~over $\{0,1\}^n$ with $\Pr[x_i = 1] = \frac12$ and $\epsilon \in [0, \frac12]$. Then,
\begin{align*}
	\Pr_{x \leftarrow \Pd_x}\left[\sum_{i=1}^n x_i \ge \frac{n}{2}(1+\epsilon)\right] \le \exp\left(-\frac{\epsilon^2 n}{6}\right) \; .
\end{align*}
\begin{proof}
	Let $m:=\left\lceil \frac{\epsilon n}{3} \right\rceil$. We have
\begin{IEEEeqnarray*}{rCl}
	\EE_{x \leftarrow \Pd_x} \left[ \left( \sum_{i=1}^n x_i \right)^m \right] & = &
	n^m \Pr_{x \leftarrow\Pd_x\atop{(i_1,\ldots,i_m) \leftarrow [n]^m}}
	\left[
		\forall j \in [m]: x_{i_j} = 1
	\right]
	\\ & = &
	n^m \prod_{j=1}^m \Pr_{x \leftarrow\Pd_x\atop{(i_1,\ldots,i_m) \leftarrow [n]^m}}
	\left[
	x_{i_j} = 1
	\mid \forall k < j: x_{i_k} = 1
	\right]
	\\ & \le &
	n^m  \left( \frac{\epsilon}{3} \cdot 1 +  \left( 1-\frac{\epsilon}{3} \right) \cdot \frac12 \right)^m
	=
	\left(\frac{n}{2}\right)^m \left(1+\frac{\epsilon}{3}\right)^m \; .
\end{IEEEeqnarray*}
	Using Markov's inequality and $\frac{1+\epsilon/3}{1+\epsilon} \le \exp\left(-\frac{\epsilon}{2}\right)$
	for $\epsilon \in [0,\frac12]$,
	\begin{align*}
	\Pr \left[\left( \sum_{i=1}^n x_i \right)^m \ge \left(\frac{n}{2}\right)^m (1+\epsilon)^m \right]
	\le \left(\frac{1+\frac{\epsilon}{3}}{1+\epsilon}\right)^m \le \exp\left(-\frac{\epsilon^2 n}{6}\right)
	\; .
	\end{align*}
	
\end{proof}
\end{theorem}

The above is the simplest proof of the most basic
Chernoff bound we know of, and we believe that it is 
worthwhile to state it explicitly.
It can be obtained by adapting the proof given in \cite{IK10} for
the given setting, although a direct adaptation yields
a slightly different (and probably a bit longer) argument.
Alternatively, it can be seen as an instantiation of 
the proof given in \cite{JOR02} in case one is interested
in counting the number of copies of $K_2$ (i.e., the number of edges) 
in a random graph $\mathsf{G}_{n,p}$, after rather many simplifications
that can be done for this very special case.
Finally, it is a straightforward instantiation of our later proof given
in Section~\ref{sec:chernoff}.

\section{Growth Boundedness}
\label{sec:chernoff}

In this section we present the definition of growth-boundedness
and prove that it implies concentration. In Section~\ref{sec:gb-wor}
we introduce growth boundedness without
repetition: a variation of our concept that we use to prove
the expander random walk bound.

\begin{definition}\label{def:gb}
Let $\delta \ge 0$ and $m \in [n]$. A distribution
$\Pd_x$ over $x = (x_1, \ldots, x_n) \in \mathbb{R}_{\ge 0}^n$ 
with $\mu := \EE_{x \leftarrow \Pd_x\atop{i \leftarrow [n]}} [ x_i ]$ is
\emph{$(\delta, m)$-growth bounded} if
\begin{align*}
\EE_{{x\leftarrow\Pd_x}} \left[ \left(\sum_{i=1}^n x_i\right)^m \right] \leq
	(\mu n)^m (1+\delta)^m \; .
\end{align*}
\end{definition}
Equivalently, $\Pd_x$ is $(\delta, m)$-growth bounded if and only if
\begin{align*}
\EE_{{x\leftarrow\Pd_x}\atop{(i_1,\ldots,i_m)\leftarrow [n]^m}}
	\Big[ \prod_{j=1}^m x_{i_j} \Big] 
	&\leq \mu^m (1+\delta)^m \; .
\end{align*}
If random variables are over $\{0, 1\}$, this condition
reduces to
\begin{align*}
\Pr_{{x\leftarrow\Pd_x}\atop{(i_1,\ldots,i_m)\leftarrow [n]^m}}
	\Big[ \forall j \in [m]: x_{i_j} = 1 \Big] 
	&\leq \mu^m (1+\delta)^m \; .
\end{align*}

We now state our main theorem:
\begin{theorem}\label{thm:ik-bound}
Let $\Pd_x$ be a distribution over $\mathbb{R}_{\ge 0}^n$,
$\mu := \EE_{{x \leftarrow \Pd_x}\atop{i \leftarrow [n]}}[x_i]$,
$\mu > 0$, $\epsilon \ge 0$.
If $\Pd_x$ is $(\delta, m)$-growth bounded,
then
\begin{align*}
\Pr_{x\leftarrow \Pd_{x}}\Bigl[\sum_{i=1}^{n} x_{i} \geq 
\mu n(1+\epsilon) \Bigr] 
\leq 
\Bigl(\frac{1+\delta}{1+\epsilon}\Bigr)^m \nonumber
\;.
\end{align*}
\end{theorem}

\begin{proof}
By Markov's inequality and growth boundedness of $\Pd_x$,
\begin{align*}
	\Pr_{x\leftarrow \Pd_x} \Big[ \sum_{i=1}^n x_i \geq \mu n (1+\epsilon) \Big] & = 
		\Pr_{x \leftarrow \Pd_x} \Big[ \big( \sum_{i=1}^n x_i \big)^m \geq (\mu n)^m (1+\epsilon)^m \Big] \\
	& \leq \Big( \frac{1+\delta}{1+\epsilon} \Big)^m \; .
\end{align*}
\end{proof}

There is an interesting connection between this proof (inspired by \cite{JOR02})
and the one used in \cite{IK10}, for details see Section \ref{sec:ik-vs-jor}.

We obtain more convenient bounds as a corollary:

\begin{corollary}\label{cor:gb}
Let $\epsilon \ge 0$ and $\Pd_x$ be an $(\frac{\epsilon}{3}, m)$-growth bounded distribution
over $\mathbb{R}_{\ge 0}^n$ with 
$\mu := \EE_{x \leftarrow \Pd_x\atop{i \leftarrow [n]}} [ x_i ]$,
$\mu > 0$.
\begin{enumerate}
\item If $\epsilon \le \frac{1}{2}:$
$\displaystyle \Pr_{x\leftarrow \Pd_x}\Bigl[\sum_{i=1}^{n} x_{i} \geq \mu n(1+\epsilon) \Bigr] 
\leq 
\exp\Big(-\frac{\epsilon m}{2}\Big)\;.
$
\item If $\epsilon \geq \frac{1}{2}$:
$\displaystyle
\Pr_{x \leftarrow \Pd_x}\Bigl[ \sum_{i=1}^n x_i \geq \mu n(1+\epsilon) \Bigl] \leq \Big(\frac{4}{5}\Big)^m \;.
$
\item If $\epsilon \geq 3$:
$\displaystyle
\Pr_{x \leftarrow \Pd_x}\Bigl[\sum_{i=1}^{n} x_{i} \geq \mu n(1+\epsilon) \Bigr] 
\leq 
2^{-m}\;.
$
\end{enumerate}
\end{corollary}

\begin{proof}
(1) follows
because
$\frac{1+\epsilon/3}{1+\epsilon} \leq \exp\big(-\frac{\epsilon}{2}\big)$ for 
$\epsilon \in [0, \frac12]$, (2) since $\frac{1+\epsilon/3}{1+\epsilon} \leq \frac{4}{5}$
for $\epsilon \geq \frac12$
and (3) due to 
$\frac{1+\epsilon/3}{1+\epsilon} \leq \frac12$ for $\epsilon \geq 3$.
\end{proof}

For example, suppose that $x_1,\ldots,x_n$ are independent over $\{0,1\}^n$, 
$\Pr[x_i = 1] = \mu > 0$, and $\epsilon \in [0, \frac12]$.

Using that for each $M$ with $|M| \le \frac{\epsilon \mu n}{3}$ we have
\begin{align*}
\Pr_{x \leftarrow \Pd_x\atop{i \leftarrow [n]}} \left[ x_i = 1 \mid \forall j \in M: x_j = 1\right] = 
\left(\frac{|M|}{n} + \left(1-\frac{|M|}{n}\right)\mu\right) 
\le \frac{|M|}{n}+\mu \le \mu \left(1+\frac{\epsilon}{3} \right) \; ,
\end{align*}
we can conclude that $\Pd_x$ is $(\frac{\epsilon}{3}, \lceil \frac{\epsilon \mu n}{3} \rceil)$-growth bounded
and
\begin{align*}
\Pr_{x \leftarrow \Pd_x}\Bigl[\sum_{i=1}^n x_i \geq \mu n(1+\epsilon)\Bigr] 
 \leq \exp(-\epsilon^2 \mu n/6) \;.
\end{align*}

\subsection{Growth boundedness without repetition}
\label{sec:gb-wor}
If one looks at the process in the growth boundedness
definition as choosing a uniform $m$-tuple of indices
$(i_1, \ldots, i_m)$ (with repetition), it is possible
to make a similar argument for choosing
a uniform set of indices of size $m$ instead.
In particular, we find it convenient 
in the proof of the expander
random walk bound.

\begin{definition}
\label{def:gb-wor}
Let $\delta \ge -1$ and $m \in [n]$.
We say that a distribution $\Pd_x$ over $\{0,1\}^n$ 
with $\mu := \Pr_{x \leftarrow \Pd_x \atop{i \leftarrow [n]}} [x_i = 1]$
is \emph{$(\delta, m)$-growth bounded without repetition}
if
\begin{align*}
\Pr_{x \leftarrow \Pd_x \atop{M \leftarrow \binom{n}{m}}}
\Big[ \forall i \in M: x_i = 1 \Big]
\le \mu^m (1+\delta)^m \; .
\end{align*}
\end{definition}

\begin{theorem}
\label{thm:ik-wor}
Let $\Pd_x$ be a distribution over $\{0,1\}^n$,
$\mu := \Pr_{x \leftarrow \Pd_x \atop{i \leftarrow [n]}}[x_i = 1]$,
$\mu > 0$, $\epsilon \ge 0$, $c \in [0,1]$. If $\Pd_x$ is $(\delta, c\epsilon \mu n)$-growth bounded
without repetition then
\begin{align*}
\Pr_{x \leftarrow \Pd_x} \Big[ \sum_{i=1}^n x_i \ge \mu n(1+\epsilon) \Big]
\le \Big( \frac{1+\delta}{1+(1-c)\epsilon} \Big)^m \; ,
\end{align*}
where $m := c\epsilon\mu n$.
\end{theorem}
\begin{proof}
Set $q := \Pr[\sum_{i=1}^n x_i \ge \mu n (1+\epsilon)]$ and compute:
\begin{IEEEeqnarray*}{rCl}
\mu^m (1+\delta)^m & \ge & \Pr_{x \leftarrow \Pd_x \atop{M \leftarrow \binom{n}{m}}} 
[\forall i \in M: x_i = 1] \\
& \ge & q \Pr_{x \leftarrow \Pd_x \atop{M \leftarrow \binom{n}{m}}}
[\forall i \in M: x_i = 1 \mid \textstyle\sum_{i=1}^n x_i \ge \mu n (1+\epsilon)] \\
& \ge & q \prod_{i=0}^{m-1} \frac{\mu n (1+\epsilon) - i}{n-i} \\
& \ge & q \mu^m (1 + (1-c)\epsilon)^m \; .
\end{IEEEeqnarray*}
\end{proof}

\begin{corollary}
\label{cor:ik-wor}
Let $\epsilon \in [0, \frac{4}{5}]$ and $\Pd_x$
be a distribution over $\{0,1\}^n$ that is
$(\frac{\epsilon}{3}, m)$-growth bounded
without repetition for some $m \le \frac{\epsilon \mu n}{6}$
with $\mu := \Pr_{x \leftarrow \Pd_x\atop{i \leftarrow [n]}}[x_i = 1]$,
$\mu > 0$. Then,
\begin{align*}
	\Pr_{x \leftarrow \Pd_x}\Big[\sum_{i=1}^n x_i \ge \mu n (1+\epsilon) \Big]
	\le \exp \Big(- \frac{\epsilon m}{3}\Big) \; .
\end{align*}
\end{corollary}

\begin{proof}
Apply Theorem \ref{thm:ik-wor} and note that
$\frac{1+\epsilon/3}{1+5\epsilon/6} \le \exp\big(-\frac{\epsilon}{3}\big)$
for $\epsilon \in [0, \frac{4}{5}]$.
\end{proof}

\subsection{Connection of \cite{IK10} and \cite{JOR02}}
\label{sec:ik-vs-jor}

Recall the proof of Theorem \ref{thm:ik-bound}.
In the context of \cite{IK10} and \cite{JOR02} 
we find it instructive to give an alternative proof,
restricted to distributions over $\{0,1\}^n$ 
(essentially the same as the proof of Theorem \ref{thm:ik-wor}).

\begin{theorem}
Let $\Pd_x$ be a distribution over $\{0,1\}^n$,
$\mu := \Pr_{{x \leftarrow \Pd_x}\atop{i \leftarrow [n]}}[x_i = 1]$,
$\mu > 0$, $\epsilon \ge 0$.
If $\Pd_x$ is $(\delta, m)$-growth bounded,
then
\begin{align}
\Pr_{x\leftarrow \Pd_{x}}\Bigl[\sum_{i=1}^{n} x_{i} \geq 
\mu n(1+\epsilon) \Bigr] 
\leq 
\Bigl(\frac{1+\delta}{1+\epsilon}\Bigr)^m \nonumber
\;.
\end{align}
\end{theorem}

\begin{proof}
Set
$q := \Pr\bigl[\sum_{i=1}^{n} x_{i} \geq \mu n(1+\epsilon) \bigr]$,
and see that\footnote{Clearly $q = 0$ is not a problem.}
\begin{IEEEeqnarray*}{rCl}
\mu^m (1+\delta)^m & \geq &
	\Pr_{{x\leftarrow\Pd_x}\atop{(i_1,\ldots,i_m) \leftarrow [n]^m}}
	[\forall j \in [m]: x_{i_j}=1] \\
& \geq & q \Pr_{x \leftarrow \Pd_x\atop{(i_1, \ldots, i_m)\leftarrow [n]^m}} [\forall j \in [m]: x_{i_j}=1 \mid \textstyle\sum_{i=1}^{n} x_{i} \geq \mu n(1+\epsilon)]\\
& \geq & q \, \mu^m (1+\epsilon)^m\;.
\end{IEEEeqnarray*}
\end{proof}

The basic idea of the proof in \cite{IK10} is to consider
$\Pr_{x, M}[\forall i \in M: x_i = 1]$,
where $M$ is a subset of $[n]$ obtained by including each element
in $M$ independently with some probability $q$.
Then, this is compared with $\Pr_{x, M}[\forall i \in M: x_i = 1 \mid \mathcal{E}]$,
where $\mathcal{E}$ is the event that $\sum_{i=1}^n x_i \ge \mu n(1+\epsilon)$.
In fact, we have
\begin{align*}
\Pr_{x}[\mathcal{E}] \leq \frac{\Pr_{x, M}[\forall i \in M: x_i = 1]}%
{\Pr_{x, M}[\forall i \in M: x_i = 1 \mid \mathcal{E}]} \; .
\end{align*}
It is possible to show that for $m := \EE[|M|] \ll n$ we have
$\Pr_{M}[\forall i \in M: x_i = 1 \mid \mathcal{E}] \gtrsim \mu^m(1+\epsilon)^m$.
To see the intuition of this, simply note that 
this probability roughly equals the probability
of only selecting red balls when one chooses with repetition 
$m$ times out of $n$ balls,
at least $\mu n(1+\epsilon)$ of which are red.\footnote{The difference to the 
actual random experiment is that we do not keep each ball with probability
$m/n$ but instead choose exactly $m$ times.}
Thus, 
\begin{align}\label{eq:1}
\Pr_{x}[\mathcal{E}]  
\lesssim \frac{\Pr_{x, M}[\forall i \in M: x_i = 1]}{\mu^m(1+\epsilon)^m}\;.
\end{align}
Now note that this last argument only uses the probability
over $M$, and so is independent of the distribution of $x$.
Thus, for any distribution on which we can give a good upper bound
on $\Pr_{x, M}[\forall i \in M: x_i = 1]$, the technique of \cite{IK10}
gives a concentration result.

The argument we use is very similar, but we pick $M$ as an
$m$-tuple whose elements are picked independently with repetition.
However, then we also have
\begin{align*}
n^m \Pr_{x, M}[\forall i \in M: x_i = 1] = \EE_{x, M}[(x_1+\ldots+x_n)^m]\;.
\end{align*}
By Markov's inequality,
\begin{align*}
\Pr[\mathcal{E}] = \Pr\left[(x_1+\dots+x_n)^m \geq (\mu n(1+\epsilon))^m\right]
\leq \frac{\Pr_{x,M}[\forall i \in M: x_i = 1]}{\mu^m(1+\epsilon)^m} \;,
\end{align*}
which is almost the same as (\ref{eq:1}).

The view in (\ref{eq:1}) is the one adopted by \cite{IK10}.
Bounding the $m$-th moment and using Markov is the view adopted
in \cite{JOR02}. 
The above argument shows that these views are closely related,
and one can argue that the connection is given by growth boundedness.

\section{Random Walks on Expanders}

\paragraph*{Overview and our results}
For an introduction to expander graphs, see
\cite{HLW06} or \cite[Chapter 4]{Vadhan12}.
In short, a $\lambda$-expander is a $d$-regular undirected
graph $G$ with the second largest (in terms of absolute value)
eigenvalue of the transition matrix at most $\lambda$.

We consider a random walk on $\lambda$-expander
starting in a uniform random vertex. It is a very useful
fact in many applications that such a random walk
behaves in certain respects
very similarly to a random walk on the complete
graph.

In particular, the so called hitting property
\cite{AKS87, Kah95}
states that the probability that an $\ell$-step
random walk on a $\lambda$-expander
$G$
stays completely inside
a set $W \subseteq V := V(G)$ with $\mu := |W|/|V|$
is at most $(\mu + \lambda)^{\ell}$.
A more general version
\cite{AFWZ95} states that
for each $M \subseteq [\ell]$
the probability that a random
walk stays inside $W$ in
all steps from $M$
is at most $(\mu + 2\lambda)^{|M|}$.

Our first result,
which may be
of independent interest,
can be considered as
a randomized version of the hitting property. 
Namely, we show that, given $\epsilon > 0$,
for a relatively small \emph{random} subset
$M \subseteq [\ell]$ of size $m$
the probability that a random
walk on a $\lambda$-expander stays inside $W$
in all steps from $M$
is at most $(\mu(1 + \epsilon))^m$:
\begin{restatable}{theorem}{expanderhitting}
\label{thm:expander-hitting}
Let $G$ be a $\lambda$-expander
with a distribution $\Pd_r$ over $V^{\ell}$ representing an
$(\ell-1)$-step random walk
$r = (v_1, \ldots, v_\ell)$
(with $v_1$ being a uniform starting vertex) and $W \subseteq V$
with $\mu := |W|/|V|$. 
Let $\epsilon \ge 0$ and $m \le \min \big(\frac12, \frac{1-\lambda}{\lambda} \frac{\epsilon \mu}{2} \big) \ell$.
Then,
\begin{align*}
\Pr_{r \leftarrow \Pd_r \atop{M \leftarrow \binom{\ell}{m}}}
\Big[ \forall i \in M: v_{i} \in W \Big] 
\le
(\mu(1 + \epsilon))^m \; .
\end{align*}
\end{restatable}

Another important property of random walks
on expander graphs is the Chernoff bound
estimating the probability that the number of
times a random walk visits $W$ is far
from its expectation. 
The first Chernoff bound for expander random
walks was given by Gillman \cite{Gil98}
and the problem was treated further in
numerous works 
\cite{Kah97,Lezaud98,LP04, Hea08, Wag08, CLLM12}.

Impagliazzo and Kabanets \cite{IK10}
apply their technique to obtain a bound
for random walks on expander graphs,
but in case of deviations smaller than
$\lambda$ they lose a
factor of $\log \big(\frac{1}{\epsilon}\big)$
in the exponent. They then ask
if their technique can be modified
to avoid this loss.

We answer this question affirmatively:
using Theorem \ref{thm:expander-hitting}
we immediately obtain a bound that
matches the known ones and
does not suffer from the additional
$\log \big(\frac{1}{\epsilon}\big)$ factor
while preserving the simplicity of the proof.
\begin{restatable}{theorem}{expandermain}
\label{thm:expander-main}
Let the setting be as in Theorem \ref{thm:expander-hitting}
with $\mu > 0$.
Define $\Pd_x$ over $\{0,1\}^{\ell}$
as $x_i = 1 \iff v_i \in W$ and let $\epsilon \in [0,\frac45]$.
Then,
\begin{align*}
\Pr_{r \leftarrow \Pd_r} 
\Big[ \sum_{i=1}^{\ell} x_i \geq \mu\ell(1+\epsilon)\Big] 
\le
2\exp\Big(-\frac{(1-\lambda) \epsilon^2 \mu \ell}{18}\Big) \; .
\end{align*}
\end{restatable}

Furthermore, we demonstrate robustness of our method
by improving the exponent to 
$\frac{1-\lambda}{1+\lambda} \frac{\mu}{1-\mu} \frac{\epsilon^2 \ell}{2} + 
o(\epsilon^2) \ell$, which is optimal
for fixed $\lambda, \mu$ and $\epsilon \to 0_+$ and
$\ell \to \infty$:

\begin{restatable}{theorem}{expandertight}
\label{thm:expander-tight}
Let the setting be as in Theorem \ref{thm:expander-hitting} with $\mu \in (0,1)$. 
Define $\Pd_x$ over $\{0,1\}^{\ell}$ as $x_i = 1 \iff v_i \in W$ and let $\epsilon \in [0,\frac12]$. Then, there exists $c_{\mu}$ that depends only on $\mu$ such that
\begin{align*}
\Pr_{r \leftarrow \Pd_r} \Big[ \sum_{i=1}^{\ell} x_i \ge \mu\ell(1+\epsilon) \Big]
	\le 2\exp\Big( - \frac{1-\lambda}{1+\lambda} \cdot \frac{\mu}{1-\mu} \cdot \frac{\epsilon^2 \ell}{2}
		+ c_\mu \epsilon^3 \ln(\frac{1}{\epsilon}) \ell \Big) \; .
\end{align*}
\end{restatable}

In the following we prove Theorems \ref{thm:expander-hitting}
and \ref{thm:expander-main}.
Section \ref{sec:coupling} contains a proof of a coupling argument
used in proofs of Theorem \ref{thm:expander-main}
and Theorem \ref{thm:expander-tight}.
In Section \ref{sec:expander-constant} we prove Theorem \ref{thm:expander-tight}
and in Section \ref{sec:optimality} we address its optimality.

\paragraph*{Proofs}
First, we need a coupling argument: let $m, \ell \in \mathbb{N}, m \le \ell$
be given. We consider the distribution 
$\mathsf{D}_{m,\ell}$ defined by the following process:
\begin{itemize}
\item Pick uniformly $M \leftarrow \binom{\ell}{m}$ and let $M := \{x_1, \ldots, x_m\}$
with $x_1 < \ldots < x_m$.
\item Let $d_1 := x_1$ and $d_i := x_i - x_{i-1}$ for $i > 1$.
\end{itemize}
A bijection shows that $d = (d_1, \ldots, d_m)$ is distributed uniformly
among the $\binom{\ell}{m}$ $m$-tuples which satisfy $\sum_{i=1}^m d_i \le \ell$
and $d_i > 0$. We now couple $\mathsf{D}_{m,\ell}$
with independent random variables 
(see Section \ref{sec:coupling} for the proof):
\begin{restatable*}{theorem}{couplingsimple}
\label{cor:dist-coupling-simple}
Let $0 < m \le \frac{\ell}{2}$. There exists a distribution
over $(d_1, \ldots, d_m, e_1, \ldots, e_m)$
such that:
\begin{itemize}
\item $e_i \leq d_i$ for $1 \le i \le m$.
\item $(d_1,\ldots,d_m)$ is distributed according to $\mathsf{D}_{m,\ell}$.
\item $(e_1,\ldots,e_{m})$ are i.i.d.~with $e_i$ in $\mathbb{N}_+$
and $\Pr[e_i = k] \le \frac{2m}{\ell}$ for every $k$.
\end{itemize}
\end{restatable*}

\begin{proof}[Proof of Theorem \ref{thm:expander-hitting}]
Pick $M \leftarrow \binom{\ell}{m}$ 
and let $(d_1, \ldots, d_m)$ be 
as in the definition of $\mathsf{D}_{m, \ell}$. 

\begin{lemma}\label{lem:exp-algebra}
\begin{align*}
\Pr_{r \leftarrow \Pd_r\atop{M \leftarrow \binom{\ell}{m}}}
\big[\forall i \in M: v_i \in W \big] \leq 
\EE_{M \leftarrow \binom{\ell}{m}} 
\big[ \prod_{i=1}^m (\mu + \lambda^{d_i}) \big] \; .
\end{align*}
\end{lemma}

\begin{proof}
Let $v := (\frac1n, \ldots, \frac1n)$ be the vector of 
the uniform distribution on $V$ and
let $P_W$ be a diagonal $n \times n$ matrix with $(P_W)_{uu} = 1$
if $u \in W$ and $(P_W)_{uu} = 0$ otherwise. Note
that $P_W^2 = P_W$.

Let $A_G$ be the probability transition matrix of $G$.
Let us denote the spectral norm of a matrix
with $||\cdot||$. We bound the probability of a random walk
staying in $W$ on indices of $M$ using a standard technique.
In particular, we use
(for the proof see \cite[Claim 4.21]{Vadhan12}):
\begin{claim}\label{cl:exp-decomposition}
\begin{align*}
||P_W A_G^k P_W || \leq \mu + (1-\mu) \lambda^k \; .
\end{align*}
\end{claim}

Fix $M$. First of all, by induction (and noting that $vA_G = v$):
\begin{align*}
\Pr_{r \leftarrow \Pd_r}
[ \forall i \in M: v_{i} \in W ] = 
\big| v P_W \prod_{i=2}^m A_G^{d_i} P_W
	\big|_1 \;.
\end{align*}

Estimate:
\begin{IEEEeqnarray}{rCl}
 \big| v P_W \prod_{i=2}^m A_G^{d_i} P_W
	\big|_1
	& \leq & \sqrt{\mu n} \cdot \big|\big| v P_W \prod_{i=2}^m A_G^{d_i} P_W \big|\big| 
		\label{eq:cs} \\
	& \leq & \sqrt{\mu n} \cdot \big|\big| v P_W \big|\big| 
		\ \prod_{i=2}^m \big|\big| P_W A_G^{d_i} P_W \big|\big| 
		\label{eq:spectral_mul} 
		\IEEEeqnarraynumspace \\
	& = & \mu \prod_{i=2}^m \big|\big| P_W A_G^{d_{i}} P_W \big|\big| 
		\label{eq:s_1} \\
	& \leq & \prod_{i=1}^m (\mu + \lambda^{d_i}) \;,
		\label{eq:eigenvalue_decomp}
\end{IEEEeqnarray}
where (\ref{eq:cs}) is due to Cauchy-Schwarz inequality (note there are at most $\mu n$ non-zero 
coordinates in the final vector), (\ref{eq:spectral_mul}) follows from $||AB|| \leq ||A|| \cdot ||B||$,
(\ref{eq:s_1}) from $||v P_W|| = \sqrt{\frac{\mu}{n}}$ and
(\ref{eq:eigenvalue_decomp}) from Claim \ref{cl:exp-decomposition}.

Since the inequality holds for every $M$, it also holds on average.
\end{proof}

The hope is that $(d_1, \ldots, d_m)$ behave ``almost''
like i.i.d.~uniform random variables. This is
indeed true, and by Theorem \ref{cor:dist-coupling-simple}
we have $(e_1, \ldots, e_m)$ such that
$e_i \le d_i$ and $e_i$ are i.i.d.~with $e_i$ in
$\mathbb{N}_+$ and $\Pr[e_i = k] \le \frac{2m}{\ell}$ 
for each $k$.

Putting this fact together with Lemma \ref{lem:exp-algebra}:
\begin{IEEEeqnarray*}{rCl}
\Pr_{r \leftarrow \Pd_r \atop{M \leftarrow \binom{\ell}{m}}}
\big[\forall i \in M: v_{i} \in W \big] 
& \le & 
\EE
\Big[ \prod_{i=1}^{m} \big(\mu + \lambda^{e_i}\big) \Big] 
\\ & = & 
\prod_{i=1}^{m} 
\big(\mu + \EE [\lambda^{e_i}] \big) \\
& \le & 
\Big( \mu + \frac{2m}{\ell} \cdot \frac{\lambda}{1-\lambda} 
	\Big)^m 
\le 
\mu^m (1+\epsilon)^m \; .
\end{IEEEeqnarray*}
\end{proof}

An immediate corollary of Theorem \ref{thm:expander-hitting} is:
\begin{corollary}
\label{cor:expander-gb}
Let the setting be as in Theorem \ref{thm:expander-hitting}.
Define $\Pd_x$ over $\{0,1\}^{\ell}$
as $x_i = 1 \iff v_i \in W$.
Then, $\Pd_x$ is $\left(
\epsilon, \left\lfloor \min \left( \frac{\ell}{2}, \frac{1-\lambda}{\lambda} \frac{\epsilon \mu \ell}{2} \right) \right\rfloor 
\right)$-growth
bounded without repetition.
\end{corollary}

\begin{proof}[Proof of Theorem \ref{thm:expander-main}]:
Combine Corollary \ref{cor:expander-gb} with Corollary
\ref{cor:ik-wor} (setting $m := \lfloor \frac{(1-\lambda)\epsilon\mu\ell}{6} \rfloor$).
\end{proof}

\subsection{Expander random walk coupling argument}
\label{sec:coupling}

Let $m, \ell \in \mathbb{N}$, $m \leq \ell$ be given.
We consider the distribution $\mathsf{D}_{m,\ell}$ defined by the following process:
\begin{itemize}
\item Pick uniformly $M \leftarrow \binom{\ell}{m}$ and let $M := \{x_1, \ldots, x_m\}$
with $x_1 < \ldots < x_m$.
\item Let $d_1 := x_1$ and $d_i := x_i - x_{i-1}$ for $i > 1$.
\end{itemize}
A bijection shows that $d = (d_1, \ldots, d_m)$ is distributed uniformly
among the $\binom{\ell}{m}$ $m$-tuples which satisfy $\sum_{i=1}^m d_i \le \ell$
and $d_i > 0$. We will now couple $\mathsf{D}_{m,\ell}$
with independent random variables.

The following two claims are proven by indicating
a straightforward bijection:

\begin{claim}\label{cl:pick-d-eq}
Conditioned on $d_1 = k$ (with $k+m-1\le \ell$),
$d' = (d_2, \ldots, d_m)$ is distributed according
to $\mathsf{D}_{m-1,\ell-k}$.
\end{claim}

\begin{claim}\label{cl:pick-d-gt}
Conditioned on $d_1 > k$ (with $k+m \le \ell$),
$d' = (d_2, \ldots, d_m, d_1-k)$ is distributed according
to $\mathsf{D}_{m,\ell-k}$.
\end{claim}

\begin{lemma}
\label{lem:ind-coupling}
Let $1 \le m^* \le m \le \ell$, and $\alpha, \beta \in \mathbb{R}$ such
that $1 \le \alpha \le \frac{\ell}{m+m^*}$ and $\frac{m}{\ell-\alpha m^*} \le \beta \le 
\frac{1}{\alpha}$ 
be given.
Then there exists a distribution
over $(d_1,\ldots,d_m, e_1,\ldots, e_{m^*})$
such that:
\begin{itemize}
\item $e_i \leq d_i$ for $1 \leq i \leq m^*$.
\item $(d_1,\ldots,d_m)$ is distributed according to $\mathsf{D}_{m,\ell}$.
\item $(e_1,\ldots,e_{m^*})$ are i.i.d.~with $e_i$ in $\mathbb{Z}_+$
and $\Pr[e_i = k] = \beta$ for $k \le \alpha$.
\end{itemize}
\end{lemma}
\begin{proof}
Induction on $m^*$.

First, pick $d_1$ and $e_1$ together with properties as stated. This is possible,
since, by union bound, 
$\Pr[d_1 \le k] \le \frac{km}{\ell}$,
and, furthermore,
$\frac{m}{\ell} \le \frac{m}{\ell-\alpha m^*} \le \beta$ and
$\lfloor \alpha \rfloor\beta \le \alpha \beta \le 1$.
What is more, ensure that $e_1$ is always at most
$\lfloor \alpha \rfloor + 1$.

If $d_1 \le \alpha$, pick $(d_2, \ldots, d_m, e_2, \ldots, e_{m^*})$
from $\mathsf{D}_{m-1,\ell-d_1}$ 
using the inductive assumption, Claim \ref{cl:pick-d-eq}, 
$\alpha(m+m^*-2) \le l-d_1$ and $\frac{m-1}{(\ell-d_1)-\alpha(m^*-1)} \le \frac{m}{l-\alpha m^*}$.

If $d_1 > \alpha$, discard it and 
pick $(d_2, \ldots, d_m, d_1-\lfloor \alpha \rfloor, e_2, \ldots, e_{m^*})$
from $\mathsf{D}_{m, \ell-\lfloor \alpha \rfloor}$
using the inductive assumption, Claim \ref{cl:pick-d-gt},
$\alpha(m+m^*-1) \le l-\lfloor \alpha \rfloor$ and 
$\frac{m}{(\ell-\lfloor \alpha \rfloor)-\alpha(m^*-1)} \le \frac{m}{\ell-\alpha m^*}$.
Note that $e_1 \le d_1$ still holds.
\end{proof}

Setting $m := m^*$ and $\beta := \frac{m}{\ell - \alpha m}$ we get as a corollary:
\begin{theorem}\label{thm:dist-coupling}
Let $0 < m \le \ell$ and $1 \le \alpha \le \frac{\ell}{2m}$. There exists a distribution
over $(d_1, \ldots, d_m, e_1, \ldots, e_m)$
such that:
\begin{itemize}
\item $e_i \leq d_i$ for $1 \le i \le m$.
\item $(d_1,\ldots,d_m)$ is distributed according to $\mathsf{D}_{m,\ell}$.
\item $(e_1,\ldots,e_{m})$ are i.i.d.~with $e_i$ in $\mathbb{Z}_+$
and $\Pr[e_i = k] = \frac{m}{\ell - \alpha m}$ for $k \le \alpha$.
\end{itemize}
\end{theorem}

For a simplified bound set $\alpha := \frac{\ell}{2m}$:
\couplingsimple

\subsection{Expanders -- improving the constant}
\label{sec:expander-constant}

With a more careful computation and
using a tighter version of Theorem \ref{thm:expander-hitting}
we arrive at a bound with a better exponent
when $\epsilon \to 0$:
\expandertight*

This exponent is optimal up to $o(\epsilon^2)$ and $o(\ell)$ factors
(see Section \ref{sec:optimality}).

As far as we know, the bound of this form has not been explicitly stated
previously, but it can be obtained with some work from previous results
(e.g., \cite{Kah97} or \cite{LP04}). Still, we believe our proof to be
considerably simpler and more elementary.

We proceed to the proof of Theorem \ref{thm:expander-tight}.

\begin{theorem}
\label{thm:expander-hitting-tight}
Let $G$ be a $\lambda$-expander
with distribution $\Pd_r$ over $V^{\ell}$ representing an
$(\ell-1)$-step random walk
$r = (v_1, \ldots, v_\ell)$
(with $v_1$ being a uniform starting vertex) and $W \subseteq V$
with $\mu := |W|/|V|$. 
Let $m \in [\ell]$ and $1 \le \alpha \le \frac{\ell}{2m}$.
Then,
\begin{align*}
\Pr_{r \leftarrow \Pd_r\atop{M \leftarrow \binom{\ell}{m}}}
\big[\forall i \in M: v_i \in W \big] \leq 
\Big( \mu + (1-\mu)\big(\frac{m}{\ell-\alpha m}\frac{\lambda}{1-\lambda} + \lambda^{\alpha} \big)
\Big)^m \; .
\end{align*}
\end{theorem}

\begin{proof}
As in the proof of Theorem \ref{thm:expander-hitting} we pick
$M \leftarrow \binom{\ell}{m}$ and show:
\begin{lemma}\label{lem:exp-algebra-tight}
\begin{align*}
\Pr_{r \leftarrow \Pd_r\atop{M \leftarrow \binom{\ell}{m}}}
\big[\forall i \in M: v_i \in W \big] \leq 
\EE_{M \leftarrow \binom{\ell}{m}} 
\big[ \prod_{i=1}^m (\mu + (1-\mu) \lambda^{d_i} ) \big] \; .
\end{align*}
\end{lemma}

\begin{proof}
Exactly the same as for Lemma \ref{lem:exp-algebra},
only we do not ignore the
$(1-\mu)$ factor in Claim \ref{cl:exp-decomposition}.
\end{proof}

By Theorem \ref{thm:dist-coupling} we can couple
$(d_1, \ldots, d_m)$ with i.i.d~$(e_1, \ldots, e_m)$
with
$e_i \le d_i$,
$e_i \in \big[ \lfloor \alpha \rfloor + 1 \big]$ and $\Pr[e_i = k] = \frac{m}{\ell-\alpha m}$ 
for $k \le \alpha$.

Putting this together with Lemma \ref{lem:exp-algebra-tight}:
\begin{IEEEeqnarray*}{rCl}
\Pr_{r \leftarrow \Pd_r \atop{M \leftarrow \binom{\ell}{m}}}
\big[\forall i \in M: v_{i} \in W \big] 
& \le & 
\EE
\Big[ \prod_{i=1}^{m} \big(\mu + (1-\mu) \lambda^{e_i}\big) \Big] 
\\ & = & 
\prod_{i=1}^{m} 
\Big( \mu + (1-\mu) \EE [\lambda^{e_i}]  \Big)
\\ & \le & 
\prod_{i=1}^{m} 
\Bigg( \mu + (1-\mu) \Big( \sum_{j=1}^{\lfloor \alpha \rfloor+1} \Pr[e_i = j] \lambda^j 
\Big) \Bigg) \\
& \le & 
\Bigg( \mu + (1-\mu) \Big( \frac{m}{\ell-\alpha m} 
\frac{\lambda}{1-\lambda} + 
\lambda^{\alpha} \Big) \Bigg)^m
\; .
\end{IEEEeqnarray*}
\end{proof}

\begin{proof}[Proof of Theorem \ref{thm:expander-tight}]
Set $m := \lfloor \frac{1-\lambda}{1+\lambda} \cdot \frac{\mu}{1-\mu} \cdot \epsilon \ell \rfloor$
and $\alpha := \max(1, \log_{\lambda}(\mu \epsilon^2))$. W.l.o.g.~assume that
$\epsilon \le \min(\frac{1}{3}, \mu, -\frac{1-\mu}{3 \ln (\epsilon)})$.
Note that $2m\alpha \le \ell$ indeed holds (use $-\frac{1-\lambda}{(1+\lambda) \ln(\lambda)} \le \frac{1}{2}$
for $\lambda \in [0,1)$).

Apply Theorem \ref{thm:expander-hitting-tight} to get:
\begin{IEEEeqnarray*}{rCl}
	\Pr_{r \leftarrow \Pd_r \atop{M \leftarrow \binom{\ell}{m}}} 
	\Big[ \forall i \in M: x_i = 1 \Big] & \le &
	\mu^m \Big(1 + \frac{\lambda}{1-\lambda}\frac{(1-\mu)}{\mu}\frac{m}{\ell-\alpha m} + \epsilon^2 \Big)^m \\
	& \le &
	\mu^m \Big( 1 + \frac{\lambda}{1-\lambda}\frac{(1-\mu)}{\mu}\frac{m}{\ell}\big(1 + \frac{2\alpha m}{\ell}\big) + \epsilon^2 \Big)^m
	\IEEEyesnumber \label{eq:ha1} \\
	& \le &
	\mu^m \Big( 1 + \frac{\lambda}{1+\lambda}\epsilon + \frac{3}{1-\mu} \epsilon^2 \ln(\frac{1}{\epsilon})\big) \Big)^m
	\IEEEyesnumber \label{eq:ha2}
	\; ,
\end{IEEEeqnarray*}
where in (\ref{eq:ha1}) we used $\frac{1}{1-\delta} \le 1 + 2\delta$ for $\delta \in [0, \frac12]$. On the other hand, estimate:
\begin{IEEEeqnarray*}{rCl}
\Pr_{r \leftarrow \Pd_r \atop{M \leftarrow \binom{\ell}{m}}} \Big[ \sum_{i=1}^{\ell} x_i \ge \mu\ell(1+\epsilon)
	\mid \forall i \in M: x_i = 1 \Big] & \ge &
	\mu^m \prod_{i=0}^{m-1} \frac{\ell(1+\epsilon)-\frac{i}{\mu}}{\ell-i}
	\\ & \ge & 
	\mu^m \exp \Big( \sum_{i=0}^{m-1} \ln \big( \frac{\ell(1+\epsilon)-\frac{i}{\mu}}{\ell-i} 
	\big) \Big)
	\\ & \ge &
	\mu^m \exp \Big( \int_{0}^{m} \ln \big( \frac{\ell(1+\epsilon)-\frac{x}{\mu}}{\ell-x} \big) \, \mathrm{d}x 
	 \Big) \; .
	\IEEEyesnumber \label{eq:ha3}
	\IEEEeqnarraynumspace
\end{IEEEeqnarray*} 

Since we have
\begin{IEEEeqnarray*}{rCl}
\Pr_{r \leftarrow \Pd_r} \Big[ \sum_{i=1}^{\ell} x_i \ge \mu\ell(1+\epsilon) \Big] & \le &
\frac{\Pr_{r \leftarrow \Pd_r \atop{M \leftarrow \binom{\ell}{m}}} \Big[ \forall i \in M: x_i = 1 \Big] }
{ \Pr_{r \leftarrow \Pd_r \atop{M \leftarrow \binom{\ell}{m}}} \Big[ \sum_{i=1}^{\ell} x_i \ge \mu\ell(1+\epsilon)
	\mid \forall i \in M: x_i = 1 \Big] } \; ,
\end{IEEEeqnarray*}
it is enough to lower bound the logarithm of the quotient of (\ref{eq:ha3}) and (\ref{eq:ha2}).
Using $\ln(1+\delta) \ge \delta - \frac{\delta^2}{2}$ for $\delta \ge 0$:
\begin{IEEEeqnarray*}{rCl}
\ln \Bigg(
\frac
{\mu^m \exp \Big( \int_{0}^{m} \ln \big( \frac{\ell(1+\epsilon)-\frac{x}{\mu}}{\ell-x} \big) \, \mathrm{d}x 
 \Big)} 
{\mu^m \Big( 1 + \frac{\lambda}{1+\lambda}\epsilon + \frac{3}{1-\mu} \epsilon^2 \ln(\frac{1}{\epsilon}) \Big)^m}
\Bigg)
& = &
\int_{0}^{m} \ln \big( \frac{\ell(1+\epsilon)-\frac{x}{\mu}}{\ell-x} \big) \, \mathrm{d}x
\\ & &
-\: m \ln \Big( 1 + \frac{\lambda}{1+\lambda}\epsilon + \frac{3}{1-\mu} \epsilon^2 \ln(\frac{1}{\epsilon}) \Big)
\\ &\ge &
\int_0^m \ln \Big( 1 + \epsilon - \frac{1-\mu}{\mu}\frac{x}{\ell} \Big) \, \mathrm{d}x
\\ & &
-\: m \big(\frac{\lambda}{1+\lambda}\epsilon + \frac{3}{1-\mu}\epsilon^2 \ln(\frac{1}{\epsilon}) \big)
\\ & \ge &
\int_0^m \epsilon - \frac{(1-\mu)x}{\mu\ell} -\epsilon^2 \, \mathrm{d}x
\\ & &
-\: m \big(\frac{\lambda}{1+\lambda}\epsilon + \frac{3}{1-\mu}\epsilon^2 \ln(\frac{1}{\epsilon}) \big)
\\ & \ge &
\epsilon m - \frac{(1-\mu)m^2}{2 \mu \ell} - \frac{\lambda \epsilon m}{1+\lambda}
- \frac{4}{(1-\mu)^2}\epsilon^3 \ln(\frac{1}{\epsilon}) \ell
\\ & \ge &
\frac{\lambda}{1+\lambda} \frac{\mu}{1-\mu} \frac{\epsilon^2 \ell}{2}
- \frac{4}{(1-\mu)^2}\epsilon^3 \ln(\frac{1}{\epsilon})\ell - \frac{1}{3}\; .
\end{IEEEeqnarray*}
\end{proof}

We remark that the proof gives $c_{\mu} \le \frac{4}{(1-\mu)^2}$ for
$\epsilon \le \min(\frac13, \mu, -\frac{1-\mu}{3\ln(\epsilon)})$.



\subsection{Optimality}
\label{sec:optimality}
Our bound is optimal in the following sense:
fix $\lambda, \mu \in (0,1) \cap \mathbb{Q}$
and let $G$ be any regular graph such that its
probability transition matrix $A_G = \lambda I_n + \frac{1-\lambda}{n} J_n$,
where $I_n$ is the identity matrix, $J_n$ the all-ones matrix and $n = |V(G)|$.
Let $W$ be an arbitrary subset of $V(G)$ such that $|W| = \mu n$.

It is easy to see that $G$ is a $\lambda$-expander. As previously,
consider an $\ell$-step random walk on $G$ with a uniform starting
vertex and define $\Pd_x$ over $\{0,1\}^{\ell}$ as $x_i = 1$ if and only if
the $i$-th step of the random walk is in $W$.

\begin{theorem}
\label{thm:expander-optimality}
There exist $\epsilon_{\lambda,\mu} > 0$ and $c_{\lambda,\mu} \in \mathbb{R}$
such that for every $\epsilon \in (0, \epsilon_{\lambda,\mu})$ and $\ell$
big enough (where ``big enough'' depends on $\lambda$, $\mu$ and $\epsilon$),
we have
\begin{align*}
	\Pr_{x \leftarrow \Pd_x} \left[ \sum_{i=1}^{\ell} x_i \ge \mu \ell (1+\epsilon) \right]
		\ge
	\exp \left(
		- \frac{1-\lambda}{1+\lambda} \cdot \frac{\mu}{1-\mu} \cdot \frac{\epsilon^2 \ell}{2}
		- c_{\lambda,\mu} \cdot \epsilon^3 \ell
	\right) \; .
\end{align*}
\end{theorem}

Theorem \ref{thm:expander-optimality} can be proved from optimality
results in either \cite{Kah97} or \cite{LP04}. For completeness,
we give a sketch of a direct proof (based on \cite{Kah97}).

\begin{proof}[Proof sketch of Theorem \ref{thm:expander-optimality}]
	Let $x := (1-\lambda)\mu(1-\mu) + \frac{(1-\lambda)\mu(1-2\mu)}{1+\lambda}\epsilon$.
We lower bound our probability with the probability of the following event:
there exist positive integers $(m_1, \ldots, m_{x\ell}), (n_1, \ldots, n_{x\ell})$
with $\sum_{i=1}^{x\ell} m_i = \mu \ell(1+\epsilon)$
and $\sum_{i=1}^{x\ell} n_i = \ell-\mu \ell(1+\epsilon)$
such that the random walk first spends $m_1$ steps in $W$,
then $n_1$ steps outside $W$, $m_2$ steps in $W$, $n_2$ steps outside of $W$ 
and so on.

Let $a := \lambda+\mu-\lambda \mu$ and $b := 1-\mu + \lambda \mu$.
Note that $a$ is the probability of staying in $W$ conditioned on 
most recent step being in $W$ and $b$ is the probability
of staying outside $W$ conditioned on most recent step
being outside of $W$.

Counting the number of ways in which one can pick $(m_1, \ldots, m_{x\ell})$
and $(n_1, \ldots, n_{x\ell})$ and the probability of picking each of them:
\begin{IEEEeqnarray*}{rCl}
	\Pr_{x \leftarrow \Pd_x} \left[ \sum_{i=1}^\ell x_i \ge \mu \ell (1+\epsilon) \right]
	& \ge & \\
	\IEEEeqnarraymulticol{3}{c}{
		\binom{(1+\epsilon)\mu\ell-1}{x\ell-1} \binom{\ell-(1+\epsilon)\mu\ell-1}{x\ell-1}
		a^{((1+\epsilon)\mu-x)\ell} b^{(1-(1+\epsilon)\mu-x)\ell} (1-a)^{x\ell} (1-b)^{x\ell} \; ,
	}
\end{IEEEeqnarray*}
which can be shown by a rather cumbersome computation to give the claimed result.
\end{proof}

\section{Polynomial Concentration}

In certain applications it is desired to bound the
concentration not only of the sum,
but rather of a (low-degree) polynomial of some random variables.

In the case when (informally) the polynomial is such
that the change in its value is bounded
when the value of a single input variable
is changed the Azuma's inequality can be applied to bound
concentration.

If this is not so, one can use techniques that were
invented by Kim and Vu \cite{KV00} and developed
in a body of work that followed
(in particular \cite{Vu02, SS12}).
In the special case of a multilinear low-degree polynomial $p(v)$
and an independent distribution of input variables $\Pd_v$
their concentration bound can be expressed,
very roughly speaking, as a function of $\frac{\mu_0}{\mu'}$,
where $\mu_0$ is the expectation of $p(v)$ and
$\mu' = \max_{K \neq \emptyset} \EE[\partial_K p(v)]$.

We obtain a bound in similar spirit.
It is not tight in general,
but can be applied to arbitrary polynomials
with positive coefficients 
over input random variables in $[0, 1]$
and is tight
in the case of
\emph{elementary symmetric polynomials}
$e_k(v) := \sum_{|S| = k} \prod_{i \in S} v_i$
(see Section \ref{sec:polynomial-discussion} for the proof).

Most importantly, as opposed to prior results,
it does not require
the input variables to be independent, but rather
\emph{almost independent} in a certain
sense (for simplicity we limit ourselves to
multilinear polynomials and inputs in \{0,1\}
for now,
full treatment can be found in Section \ref{sec:polynomial}): 

\begin{definition}
Let $\Pd_v$ be a distribution over $\{0,1\}^\ell$,
$\delta \ge 0$ and $m \in [\ell]$. $\Pd_v$
is $(\delta, m)$-almost independent if
for each $M \subseteq [\ell]$ with 
$|M| \le m$
\begin{align*}
\Pr_{v \leftarrow \Pd_v} [\forall i\in M: v_i = 1] \le 
(1+\delta)^m \prod_{i \in M} \Pr_{v \leftarrow \Pd_v} [v_i = 1] \; .
\end{align*}
\end{definition}

Let us state our main theorem of this section.

Let $\Pd_v$ be a $(\delta, km)$-almost independent
distribution. Let $p(v)$ be a multilinear polynomial of degree $k$ with
positive coefficients. Our way to deal with dependencies in $\Pd_v$
is to state the bound in terms of $\Pd^*_v$ which is the distribution
of independent variables with the same marginals as $\Pd_v$
(i.e., each $v^*_i$ has the same distribution as $v_i$).

We express the concentration in terms of
\begin{align*}
\mu^*_i := \max_{K \subseteq [\ell] \atop{ |K| = i}} 
 \EE_{v \leftarrow \Pd^*_v} [\partial_K p(v)] \; .
\end{align*}
Note that $\mu_0^*$ is the expectation of $p(v)$ under $\Pd^*_v$.

\begin{theorem}
\label{thm:kv-simple}
Let the setting be as above and $\epsilon > 0$. Then,
\begin{align*}
\Pr_{v \leftarrow \Pd_v} \Big[ p(v) \ge \mu^*_0 (1+\epsilon) \Big]
\le
\Big( \frac{(1+\delta)^k (1 + \frac{\sum_{i=1}^k \binom{km}{i} \mu^*_i}{\mu^*_0})}{1+\epsilon} \Big)^m \; .
\end{align*}
\end{theorem}

\begin{proof}[Proof outline]
Write $p(v)$ as a sum of binary random
variables (corresponding to the monomials) $x_1, \ldots, x_n$.
Due to Theorem \ref{thm:ik-bound} it is enough to show that $(x_1, \ldots, x_n)$ are $(\delta', m)$-growth
bounded, where 
$1 + \delta' = (1+\delta)^k\big(1 + \frac{\sum_{i=1}^k \binom{km}{i}\mu_i^*}{\mu_0^*}\big)\frac{\mu_0^*}{\mu}$.

Since $\Pd_v$ is $(\delta, km)$-almost independent, this task can be further reduced to
showing that if $v$ is distributed according to $\Pd_v^*$ instead of $\Pd_v$, then
$(x_1, \ldots, x_n)$ are $(\delta'', m)$-growth bounded,
where $1+\delta'' = \big( 1+ \frac{\sum_{i=1}^k \binom{km}{i} \mu_i^*}{\mu_0^*} \big)$.

Fix $s < m$ and $(i_1, \ldots, i_{s}) \in [n]^{s}$
and let $M$ be the set of all indices $j$
such that $v_j$ influences at least
one of $x_{i_1}, \ldots, x_{i_{s}}$
(note that $|M| \le km$).

We write $p(v) = \sum_{K \subseteq M: |K| \le k} p_K(v)$,
where $p_K(v)$ consists of those monomials
whose variables intersected with $M$ are exactly $K$.
Observe that
\begin{align*}
	\EE_{v \leftarrow \Pd_v^*} \Big[ p_K(v) \mid \forall i \in M: v_i = 1 \Big]
	\le
	\EE_{v \leftarrow \Pd_v^*} \Big[ \partial_K p(v) \Big] \; .
\end{align*}

To get growth boundedness for $x_1, \ldots, x_n$ we proceed by induction
and bound
\begin{align*}
	\Pr_{v \leftarrow \Pd_v^* \atop{i_{s+1} \leftarrow [n]}} 
	\Big[ x_{i_{s+1}} = 1 \mid \forall j \in [s]: x_{i_j} = 1 \Big]
	& = \frac{1}{n}
	\EE_{v \leftarrow \Pd_v^*} \Big[ p(v) \mid \forall i \in M: v_i = 1 \Big]
	\\ & \le \frac{1}{n} \sum_{K \subseteq M : |K| \le k}
	\EE_{v \leftarrow \Pd_v^*} \Big[ \partial_K p(v) \Big]	
	\\ & \le \frac{\mu_0^*}{n} \Big(
	1 + \frac{\sum_{i=1}^k \binom{km}{i} \mu_i^*}{\mu_0^*}
	\Big) \; .
\end{align*}
\end{proof}

Let $\mu' := \max_{i \in [k]} \mu_i^*$. Since $\sum_{i=1}^k \binom{km}{i} \le (km)^k$,
we have:
\begin{corollary}
\label{cor:kv-simple}
Let the setting be as in Theorem \ref{thm:kv-simple}. Then,
\begin{align*}
	\Pr_{v \leftarrow \Pd_v} \left[ p(v) \ge \mu_0^* (1+\epsilon) \right]
	\le
	\Big( \frac{(1+\delta)^k (1 + \frac{(km)^k \mu'}{\mu_0^*})}{1+\epsilon} \Big)^m \; .
\end{align*}
\end{corollary}

\subsection{An application in \cite{HS12}}
\label{sec:hs12}

In \cite{HS12} the authors prove a lower bound
on the complexity of a black-box construction
of a pseudorandom generator from a one-way
function. 

Part of their proof consists in using Theorem \ref{thm:kv-simple}
to show a concentration bound for a certain polynomial.
The proof of Theorem \ref{thm:kv-simple}
is not included in \cite{HS12}, but
deferred to this paper instead.
Since the input variables of the polynomial
are not independent, to the best of our knowledge no previous work 
is applicable to this case.\footnote{
It was pointed out to us that a generalisation of the result 
of Latała and Łochowski \cite{LL03} might be applicable
(together with \cite{PMS94}). However, moment
bound in \cite{LL03} is optimal only up to a constant
in the exponent that depends on the degree
and the degree is non-constant in our setting.
}

The following random process is considered: 
pick a permutation $f: \{0,1\}^n \to \{0,1\}^n$
u.a.r.~and consider the distribution $\Pd_g$ over
$2^{2n}$ random variables $g := \left\{g_{x,y}: 
x, y \in \{0,1\}^n \right\}$  defined
as $g_{x,y} = 1$ if $f(x) = y$ and $g_{x,y} = 0$ otherwise.

The random variables in $g$ are not independent, but it
is easy to check that they are $(1, 2^{n-1})$-almost independent.
Also, the corresponding independent distribution
$\Pd^*_g$ has expectation $2^{-n}$ for each $g_{x,y}$. 

Fix $k \le \frac{n}{100 \log n}$. \cite{HS12}
defines a certain multilinear polynomial $p(g)$ of
degree at most $k$ such that $\mu_0^* \le 2^{n/15}$
and $\mu' \le 2^{n/15}$ (we omit the details).

\cite{HS12} needs to show that (for $n$ big enough):
\begin{align*}
	\Pr_{g \leftarrow \Pd_g} \left[
		p(g) \ge 2^{n/10}
	\right] \le 2^{-2^{n/100k}} \; .
\end{align*}
To this end, calculate using Corollary \ref{cor:kv-simple}
and setting $\delta:= 1$, $\epsilon := 2^{9n/100}/\mu^*_0$
and $m := 2^{n/100k}$:
\begin{IEEEeqnarray*}{rCl}
	\Pr_{g \leftarrow \Pd_g} \left[
		p(g) \ge \mu_0^* + 2^{9n/100}
	\right] & \le &
	\left(
		\frac{2^k \max\left(2, \frac{2 k^k 2^{n/100} \mu'}{\mu_0^*}\right)}
		{\frac{2^{9n/100}}{\mu_0^*}} \;
	\right)^{2^{n/100k}} \\
	& \le &
	\left(
		\frac{
			2^{k+1} \max\left(
				\mu_0^*, k^k 2^{n/100} \mu'
			\right)
		}{2^{9n/100}}
	\right)^{2^{n/100k}} \\
	& \le &
		2^{-2^{n/100k}}
	\; .
\end{IEEEeqnarray*}

\subsection{Other applications}

We note that despite the fact that the deviation for which we applied
our theorem in Section \ref{sec:hs12} is big relative to the expectation,
one can obtain meaningful bounds also for very small deviations.

This can be seen by taking a restricted version of Theorem \ref{thm:kv-simple}:

\begin{theorem} \label{thm:KV-style-simple}
Let $\Pd_v$ be a distribution of independent variables
(i.e., $\Pd_v = \Pd^*_v$) over $[0, 1]^\ell$.
Let $p(v)$ be as in Theorem \ref{thm:kv-simple} and $\epsilon \in [0, \frac{1}{2}]$.
Then:
\begin{align*}
	\Pr_{v \leftarrow \Pd_v} \Big[ p(v) \geq \mu(1 + \epsilon) \Big] \leq 2\exp\Big(-\frac{\epsilon}{6k}
		\Big( \frac{\epsilon \mu}{\mu'} \Big)^{1/k} \Big) \; .
\end{align*}
\end{theorem}
\begin{proof}
Note that $\Pd_v$ are $(0, \ell)$-almost independent.
Take $m: = \Big\lfloor \frac{1}{k} \Big( \frac{\epsilon \mu}{3\mu'} \Big)^{1/k} \Big\rfloor$,
obtain $(\frac{\epsilon}{3}, m)$-growth boundedness as in Corollary \ref{cor:kv-simple} 
and apply Corollary \ref{cor:gb}.1.
\end{proof}

For example, in a representative setting when Azuma-like methods fail:
consider the polynomial that counts the triangles
in Erdős–Rényi random graph $\mathsf{G}_{n, n^{-3/4}}$, i.e.,
$p(v) = \sum_{\{a,b,c\} \in \binom{n}{3}} v_{ab} v_{ac} v_{bc}$.
We compute $\mu = \Theta(n^{3/4})$ and $\mu' = \Theta(1)$.

For $\epsilon \in [0, \frac{3}{16}]$ Theorem \ref{thm:KV-style-simple}
gives:
\begin{align*}
\Pr_{v \leftarrow \Pd_v} \Big[ p(v) \geq \mu(1 + n^{-\epsilon}) \Big] 
\le 
\exp(-\Omega(n^{1/4 - 4\epsilon/3})) \; .
\end{align*}
This is comparable to the bound from
\cite{KV00} (which was the first paper to
give a good bound in this setting). 
Better bounds are known, 
in particular we revisit the triangle counting
in Section \ref{sec:graph}.

\subsection{Polynomial concentration -- full proof} \label{sec:polynomial}
In this section we prove our polynomial concentration theorem in the general
case (i.e., random variables in $[0,1]$ and non-linear polynomials).
For this we generalise the notion of almost independence. 

\begin{definition}\label{def:ai}
Let $\Pd_{v}$ be a distribution over $v = (v_1,\ldots,v_{\ell}) \in [0,1]^\ell$.
Given a tuple $(i_1, \ldots, i_s) \in [\ell]^s$
define $(c_1, \ldots, c_{\ell})$ as
$c_j := |\{k \in [s]: i_k = j\}|$.

Let $\delta \ge 0$ and $m \in [\ell]$.
We say that $\Pd_v$ is
\emph{$(\delta, m)$-almost independent}
if for all $(i_1, \ldots, i_s) \in [\ell]^s$ with 
$s \leq m$:
\begin{align*}
\EE_{v\leftarrow \Pd_v}
\Big[ \prod_{j=1}^s v_{i_j} \Big] \le (1+\delta)^m \prod_{j=1}^{\ell} \EE_{v\leftarrow \Pd_v} [v_i^{c_i}] \;.
\end{align*}
\end{definition}
Note that an $\ell$-wise independent distribution is $(0, \ell)$-almost independent.
As expected, for binary distributions the condition from
Definition \ref{def:ai} reduces to
\begin{align*}
\Pr_{v \leftarrow \Pd_v} [\forall i\in M: v_i = 1] \le 
(1+\delta)^m \prod_{i \in M} \Pr_{v \leftarrow \Pd_v} [v_i = 1] \; 
\end{align*}
for all sets $M \subseteq [\ell]$ with $|M| \le m$.

Let multisets $e_1,\ldots,e_n$ with elements from $[\ell]$
be given.
We define random variables $x_1, \ldots, x_n$ as 
$x_i = w_i \prod_{j \in e_i} v_{j}$ with $w_i \ge 0$ and then consider
the polynomial $p(v) := \sum_{i=1}^{n} x_i$.
We are interested in bounding the upper tail of $p(v)$.


Given a distribution $\Pd_{v}$ on $[0,1]^\ell$, let
$\Pd^*_{v}$
be the distribution with the same marginals as $\Pd_{v}$, but
in which variables $v_i$ are independent.

For $K \subseteq [\ell]$, let:
\begin{align*}
\Delta_K p(v) :=
\sum_{i \in [n]:  \forall j \in K\, v_j \in e_i} x_i |_{v_j = 1: j \in K} \; .
\end{align*}
In other words, $\Delta_K p(v)$ consists of
monomials which contain at least one copy
of each variable from $K$ with variables from $K$
set to $1$ in those monomials. Note that
in multilinear case this expression
coincides with $\partial_K p(v)$.

Inspired by \cite{KV00}, 
we let 
$\mu := \EE_{v \leftarrow \Pd_v} [p(v)]$ and
$\mu^*_i := \max_{K \subseteq [\ell], |K| = i}
	\EE_{v \leftarrow \Pd^*_v}[\Delta_K p(v)]$.
Note that $\mu^*_0 = \EE_{v \leftarrow \Pd^*_v} [p(v)]$.

\begin{theorem}\label{thm:KimVuStyle}
Let $\Pd_{v}$ be
a $(\delta,km)$-almost independent distribution
over $[0,1]^{\ell}$.  
Let $p(v)$ be as above of degree at most $k$, i.e., 
$p(v) = \sum_{i=1}^{n} x_i$ with $x_i = w_i \prod_{j\in e_i} v_j$, where 
$w_i \ge 0$ and the total cardinality of $e_i$ is at most $k$.

Then, if $\mu > 0$, for all $\epsilon > 0$ we have:
\begin{align*}
\Pr_{v \leftarrow \Pd_{v}}\Bigl[p(v) \geq \mu^*_0 (1+\epsilon)\Bigr] 
\leq 
\Bigl(\frac{(1+\delta)^k (1+\frac{\sum_{i=1}^k \binom{km}{i} \mu^*_i}{\mu_0^*})}{1+\epsilon}\Bigr)^m\;.
\end{align*}
\end{theorem}

\begin{proof} Immediately from the following lemma and Theorem \ref{thm:ik-bound}:

\begin{lemma}\label{lem:independenceOfMonomials}
The random variables $(x_1,\ldots,x_{n})$ are 
$(\delta' ,m)$-growth bounded, 
where $1+\delta' = (1+\delta)^k(1+\frac{\sum_{i=1}^k \binom{km}{i} \mu^*_i}{\mu_0^*})\frac{\mu_0^*}{\mu}$.
\end{lemma}
\begin{proof}
For each $(i_1, \ldots, i_m) \in [n]^m$:
\begin{align}
\EE_{{v \leftarrow \Pd_v}}
\Big[ \prod_{j=1}^m x_{i_j} \Big]
& \le
(1+\delta)^{km}
\EE_{v\leftarrow \Pd^*_v} \Big[ \prod_{j=1}^m x_{i_j} \Big] \;,
\label{eq:c1}
\end{align}
where we used that the $v_i$ are $(\delta, km)$-almost independent.
Therefore it is enough to show
\begin{align}
\EE_{v \leftarrow \Pd^*_v\atop{(i_1, \ldots, i_m) \leftarrow [n]^m}}
\Big[ \prod_{j=1}^m x_{i_j} \Big]
\le
\Big( 1 + \frac{\sum_{i=1}^k \binom{km}{i} \mu^*_i}{\mu^*_0} \Big)^m
\Big( \frac{\mu^*_0}{n} \Big)^m 
\;.
\label{eq:c2}
\end{align}

We proceed by induction: $m = 0$ is self-evident. For $m > 0$
and fixed $(i_1, \ldots, i_{m-1})$ we define
a set\footnote{We ``collapse'' multisets to a set $M$ in a natural way here. 
The same applies to the definition of $p_K(v)$.} 
$M := \cup_{j=1}^{m-1} e_{i_j}$, i.e., $M$ consists of all
$v_i$ that influence $(x_{i_1}, \ldots, x_{i_{m-1}})$.

For any $K \subseteq M$ with $|K| \leq k$ we let $p_K(v)$ be 
the sum over those monomials which have exactly intersection $K$ 
with $M$, i.e.,
\begin{align*}
p_K(v) := \sum_{i: e_i \cap M = K} x_i \;.
\end{align*}
Then, since $p(v) = \sum_{K: |K| \leq k} p_K(v)$ we have:
\begin{align*}
\EE_{v\leftarrow \Pd^*_v\atop{i_m \leftarrow [n]}} 
\Big[ \prod_{j=1}^m x_{i_j} \Big]
&=
\frac{1}{n} \EE_{v \leftarrow \Pd_v^*}
\Big[ \Big( \sum_{K: |K| \le k} p_K(v) \Big) \Big( \prod_{j=1}^{m-1} x_{i_j} \Big) \Big]
\\ &\le
\frac{1}{n} \EE_{v \leftarrow \Pd_v^*}
\Big[ \Big( \sum_{K: |K| \le k} \Delta_K p(v) \Big) \Big( \prod_{j=1}^{m-1} x_{i_j} \Big) \Big]
\\ &\le
\frac{1}{n} \EE_{v \leftarrow \Pd_v^*}
\Big[ \prod_{j=1}^{m-1} x_{i_j} \Big]
\sum_{K:|K| \le k} \EE_{v \leftarrow \Pd_v^*} [ \Delta_K p(v) ]
\\ &\le
\Big( 1 + \frac{\sum_{i=1}^k \binom{km}{i} \mu_i^*}{\mu_0^*} \Big)
\frac{
\mu_0^*}{n}
\EE_{v \leftarrow \Pd_v^*}
\Big[ \prod_{j=1}^{m-1} x_{i_j} \Big] \; .
\end{align*}
The inductive argument follows by averaging over all $(i_1, \ldots, i_{m-1})$.
\end{proof}
\noqed
\end{proof}

\subsection{Tightness for elementary symmetric polynomials}
\label{sec:polynomial-discussion}

We show that 
Theorem \ref{thm:KimVuStyle} is essentially tight for
\emph{elementary symmetric polynomials}
$e_k(v) := \sum_{|S| = k} \prod_{i \in S} v_i$.
For the upper bound we have:
\begin{lemma}
Fix $k \in \mathbb{N}$. Let $\epsilon \in [0, \frac{1}{2}]$, and
let $\Pd_v$ be a distribution of i.i.d.~random variables over
$\{0, 1\}^{\ell}$ with $\Pr_{v \leftarrow \Pd_v} [v_i = 1] = p > 0$.

There exists $c_k > 0$ (depending only on $k$) such that:
\begin{align*}
	\Pr_{v \leftarrow \Pd_v} [e_k(v) \geq p^k\binom{n}{k}(1 + \epsilon)] \leq \exp(-c_k \epsilon^2 p\ell) \; .
\end{align*}
\end{lemma}
\begin{proof}
We have $\mu_i \leq (p\ell)^{k-i}$ for every $i$.
What is more, there exists $c'_k$ such that $\mu \geq c'_k (p\ell)^k$.
Now apply Lemma \ref{lem:independenceOfMonomials} and Theorem \ref{thm:ik-bound}.2
for $m := c''_k \epsilon p \ell$ (again observing that $\Pd_v$ is $(0,\ell)$-almost independent). 
\end{proof}

For the lower bound, we first state a well-known tightness of the Chernoff bound
for independent coin tosses (for the proof see \cite{You12} or, alternatively, Appendix B of \cite{HR11}):
\begin{lemma} \label{lem:reverse-chernoff}
Let $\epsilon \in (0, \frac{1}{2}]$ and $\Pd_v$ be a distribution of
i.i.d.~random variables over $\{0,1\}^{\ell}$
with $\Pr_{v \leftarrow \Pd_v} [v_i = 1] = p \le \frac12$.
Furthermore, assume that $\epsilon^2p\ell \geq 3$. Then:
\begin{align*}
	\Pr_{v \leftarrow \Pd_v} \big[\sum_{i=1}^n v_i \geq p\ell(1+\epsilon)\big] \geq \exp(-9\epsilon^2p\ell) \;.
\end{align*}
\end{lemma}

\begin{lemma}
Let $k \in \mathbb{N}$, $\epsilon \in (0, \frac{1}{4}]$ and $\Pd_v$ be a distribution of 
i.i.d.~random variables over $\{0,1\}^{\ell}$
with $\Pr_{v \leftarrow \Pd_v} [v_i = 1] = p \le \frac12$.
Furthermore, assume that $\epsilon p\ell \geq k$ and $\epsilon^2 p\ell \geq \frac{3}{4}$.
Then:
\begin{align*}
	\Pr_{v \leftarrow \Pd_v} [e_k(v) \geq p^k \binom{\ell}{k} (1 + \epsilon)] \geq \exp(-36 \epsilon^2 p\ell) \; .
\end{align*}
\end{lemma}
\begin{proof}
\begin{IEEEeqnarray*}{rCl}
	\Pr \Big[e_k(v) \geq p^k \binom{\ell}{k} (1 + \epsilon) \Big] & \geq &
		\Pr \Big[e_k(v) \geq \frac{(p\ell(1+\epsilon))^k}{k!} \Big] \\
	& \geq & \Pr \Big[ e_k(v) \geq \frac{(p\ell(1+2\epsilon)-k)^k}{k!} \Big] \\
	& \geq & \Pr \Big[ e_k(v) \geq \binom{p\ell(1+2\epsilon)}{k} \Big] \\
	& = & \Pr \Big[ \sum_{i=1}^n v_i \geq p\ell(1+2\epsilon) \Big] \\
	& \geq & \exp(-36\epsilon^2p\ell) \; , \IEEEyesnumber \label{eq:j1}
\end{IEEEeqnarray*}
where (\ref{eq:j1}) follows from Lemma \ref{lem:reverse-chernoff}.
\end{proof}

\section{Counting Subgraphs in Random Graphs}
\label{sec:graph}

In the proof of the polynomial concentration bound
we consider values $\mu^*_i$ which are maxima
of expectations of $\partial_K p(v)$ over sets $K$ of size $i$.
Each such value yields a 
contribution\footnote{Think of a constant $k$ and a family of polynomials with $m$ going to infinity.} 
of $\binom{km}{i} \mu^*_i$
(proportional to the number of partial derivatives of this type
in the subset of input variables of size $km$)
and the ``quality'' of a concentration bound 
depends, roughly, 
on the maximum
such contribution.

In principle, nothing prevents us from considering a different,
possibly finer,
division of partial derivatives into a constant number
of classes, each with its own contribution.

In particular, it is an obvious fact that
the number of occurrences
of a fixed subgraph $H$ in a random Erdős–Rényi graph
(for some of the work on the problem see \cite{JR02, JOR02, JR11})
can be expressed in terms of a multilinear polynomial.
In this setting we may divide the partial derivatives into classes
corresponding to subgraphs of $H$. Interestingly,
this yields an upper tail bound proof that is 
basically isomorphic
to the famous one of Janson, Oleszkiewicz and Ruciński \cite{JOR02}.

Our result holds in the setting of almost-independent
distributions, readily applicable, for example,
to $\mathsf{G}_{n,m}$ random graphs 
(of course the proof of \cite{JOR02} also generalises to those settings).

\subsection{The proof}
We prove in our framework (a slight generalisation of) a result due to
Janson, Oleszkiewicz, and Ruciński \cite{JOR02}.

Fix $n \in \mathbb{N}$ and consider some distribution $\Pd_e$ over 
$e \in \{0, 1\}^{\binom{n}{2}}$
where we index the entries of $e$ with
$E := \{ \{u,v\} \mid u, v \in [n], u \ne v \}$,
that is the set of $\binom{n}{2}$ possible edges of $n$-vertex simple graph.
Unsurprisingly, we interpret $e_{\{u,v\}} = 1$ as the existence of
respective edge in the graph.
Let\footnote{One can modify our proof so that it works for heterogenous $p_{\{u,v\}}$,
but it is more technical than interesting.}
$p$ be such that for each $\{u, v\} \in E$
we have
$\Pr_{e \leftarrow \Pd_e} [ e_{\{u,v\}} = 1 ] \le p$.

Fix a simple graph $G = ([v_G], E_G)$ 
with $v_G$ vertices and $e_G$ edges.
We would like to count the number of (not necessarily induced)
isomorphic copies of $G$ 
in a random graph induced by $\Pd_e$.

Assume w.l.o.g.~that $G$ does not have isolated vertices.
We will only use graphs without isolated vertices
in our proof and therefore from now on we identify
a graph with the set of its edges.

We denote isomorphism of graphs by $G \sim H$.
Then the number of copies of $G$ in the graph
induced by $\Pd_e$ can
be expressed as a polynomial:
\begin{align*}
q(e) := 
\sum_{E' \subseteq E \atop{E' \sim G}} 
x_{E'}
:=
\sum_{E' \subseteq E \atop{E' \sim G}} 
\prod_{\{u,v\} \in E'} e_{\{u,v\}} 
\; ,
\end{align*}
where variables $x_{E'}$ can be thought of as
a vector $x$ distributed according to some $\Pd_x$.
The number of monomials in this sum is 
$\frac{1}{d} \prod_{i=0}^{v_G-1} (n-i)$,
where $d$ is the number of automorphisms of $G$,
and the degree of each monomial is $e_G$.

Thus, we can apply the technique from
Section \ref{sec:polynomial}. We will do it
in a more careful fashion, though, in order
to match the bound from \cite{JOR02}.

For a graph $H$ let $N(n,m,H)$ be the largest number of
copies of $H$ which can be packed into $n$ vertices and $m$ edges.
Following \cite{JOR02}, we
set:
\begin{align*}
	M^*_G(n,p) := \max \Big\{ m \le \binom{n}{2}: \forall H \subseteq E_G, H \ne \emptyset: N(n, m, H) \le n^{v_{H}} p^{e_{H}} \Big\} \; .
\end{align*}

We need the following lemma with a proof in \cite[Lemma 2.1]{JOR02}:
\begin{lemma}\label{lem:packing-inequality}
For every $H$ with $e_H > 0$ there is a constant $C_H$
such that if $n \ge v_H$ and 
$0 \le m_1 \le m_2\ \le \binom{n}{2}$,
then
\begin{align*}
	N(n, m_1, H) \le C_H \frac{m_1}{m_2} N(n, m_2, H) \; .
\end{align*}
\end{lemma}

Given $\Pd_e$, similarly as in Section \ref{sec:polynomial},
let $\mu := E_{e \leftarrow \Pd_e}[q(e)]$ and 
$\mu^* := \frac{1}{d} p^{e_G} \prod_{i=0}^{v_G-1} (n-i)$.
Note that $\mu^*$ is the expectation of $q(e)$
in the distribution where each edge appears independently
with probability $p$
(i.e., Erdős–Rényi model) and that
$E_{e \leftarrow \Pd^*_e}[q(e)] \le \mu^*$, where
$\Pd^*_e$ is the independent distribution
with the same marginals as $\Pd_e$.

\begin{lemma}\label{lem:graph-main}
Fix $\delta > 0$, as well as $n$, $\Pd_e$ and $G$. If $m$
is such that
\begin{align*}
\forall H \subseteq E_G, H\ne \emptyset: N(n, m, H) \le \frac{1}{2^{e_G} v_G^{v_G}} \delta n^{v_H} p^{e_{H}} \; ,
\end{align*}
and $\Pd_e$ is $(\delta', e_Gm)$-almost independent, then $\Pd_x$ is
$(\delta'', m)$-growth bounded, where $1 + \delta'' = (1+\delta')^{e_G}(1+\delta)\frac{\mu^*}{\mu}$.
\end{lemma}
\begin{proof}
Proceeding as in the proof of Lemma \ref{lem:independenceOfMonomials}
in (\ref{eq:c1}) and (\ref{eq:c2}), we reduce the problem to showing that
\begin{align*}
	\EE_{e \leftarrow \Pd^*_v} \big[ q(e)^m \big] \le (\mu^*)^m (1+\delta)^m \; .
\end{align*}

The rest of our argument is very similar as in \cite{JOR02},
but we give it for completeness and appreciating the connection
to the proof of Lemma \ref{lem:independenceOfMonomials}.

We proceed by induction on $m$, with $m = 0$ being a trivial case.
For $m > 0$ fix a tuple $(x_{E'_1}, \ldots, x_{E'_{m-1}})$,
with $E' := \cup_{i=1}^{m-1} E'_{i}$.

For an $H \subseteq E_G$ we define:
\begin{align*}
q_{H}(e) := 
\sum_{\substack{
	E'' \subseteq E \\
	E'' \sim G \\
	(E'' \cap E') \sim H
	}} 
	x_{E''} \; ,
\end{align*}
that is $q_H(e)$ groups all those possible occurences
of $G$ for which their intersection with $E'$
is isomorphic to $H$. Clearly 
$q(e) \le \sum_{H \subseteq E_G} q_{H}(e)$.

Define an event $\mathcal{A}$ as $\forall \{u, v\} \in E': e_{\{u,v\} = 1}$.
We have:
\begin{align*}
\EE_{e \leftarrow P^*_e} \Big[ q(e) \prod_{i=1}^{m-1} x_{E'_i} \Big]
& \le
\EE_{e \leftarrow P^*_e} \Big[ \prod_{i=1}^{m-1} x_{E'_i} \Big]
\sum_{H \subseteq E_G}
\EE_{e \leftarrow P^*_e} \Big[ q_{H}(e) \mid \mathcal{A} \Big] \; .
\end{align*}
But for $H \ne \emptyset$:
\begin{align}
\EE_{e \leftarrow P^*_e} \Big[ q_{H}(e) \mid \mathcal{A} \Big]
& =
p^{e_G - e_H} 
\cdot
\big| \big\{ E'' \subseteq E: E'' \sim G \land (E'' \cap E') \sim H \big\} \big|
\nonumber
\\ & \le
p^{e_G - e_H}
N(n, m, H)\, n^{v_G-v_H} \frac{v_G!}{d}
\label{eq:d1}
\\ & \le
\frac{v_G!}{d 2^{e_G} v_G^{v_G}} \delta n^{v_G} p^{e_G} \le \frac{\delta \mu^*}{2^{e_G}}
\nonumber
\; ,
\end{align}
where (\ref{eq:d1}) follows since each copy of $G$ corresponding
to a monomial in $p_H$ can be recovered from its
intersection with $E'$ (isomorphic to $H$), its vertices outside $E'$
and its isomorphism with $G$ (where factor $d$ accounts for the isomorphisms 
that result in the same graph). Summing over all $H$,
\begin{align*}
\sum_{H \subseteq E_G}
\EE_{e \leftarrow P^*_e} \Big[ q_{H}(e) \mid \mathcal{A} \Big]
\le \mu^* +
\sum_{H \subseteq E_G \atop{H \ne \emptyset}}
\frac{\delta \mu^*}{2^{e_G}} \le \mu^* (1+\delta)
\; .
\end{align*}

Since the choice of $(x_{E'_m}, \ldots, x_{E'_{m-1}})$ was arbitrary,
the induction follows by averaging over all such choices.
\end{proof}

\begin{theorem}\label{thm:graph-gb}
Fix $n$, $G$, $\Pd_e$, and $\delta > 0$.
There exists $C_G > 0$ depending only on $G$ such that 
If $ \frac{C_G m}{\delta} \le M_G^*(n, p)$
and $\Pd_e$ is $(\delta', e_Gm)$-almost independent, then $\Pd_x$ is
$(\delta'', m)$-growth bounded, where 
$1 + \delta'' = (1+\delta')^{e_G}(1+\delta)\frac{\mu^*}{\mu}$.
\end{theorem}
\begin{proof}
From Lemma \ref{lem:packing-inequality} and Lemma \ref{lem:graph-main}.
\end{proof}

\begin{theorem}[\cite{JOR02}]\label{thm:graph-jor}
Fix $G$ and $\epsilon \in [0, \frac12]$.
Let $\mathsf{G}_{n,p}$ be Erdős–Rényi distribution
with $n \ge v_G$ and $p > 0$.
There exists $c_G > 0$ depending only on $G$
such that:
\begin{align*}
\Pr_{e \leftarrow \mathsf{G}_{n,p}} \Big[ q(e) 
\ge \mu(1+\epsilon) \Big]
\le
\exp(-c_G \epsilon^2 M^*_G(n, p)) \; .
\end{align*}
\end{theorem}
\begin{proof}
From Theorem \ref{thm:graph-gb} and Corollary \ref{cor:gb}.1
taking $m := c'_G \epsilon M^*_G(n, p)$ for appropriately
small $c'_G$ and noting that $\mathsf{G}_{n, p}$
is $(0, \binom{n}{2})$-almost independent
and $\mu^* = \mu$.
\end{proof}

We can apply almost-independence
to the distribution $\mathsf{G}_{n,m}$ of a uniform
random graph on $n$ vertices and $m$ edges.
\begin{theorem}
Fix $G$ and $\epsilon \in [0, 1]$.
Let $\mathsf{G}_{n,m}$ be uniform distribution
on graphs with $n$ vertices and $m$ edges
with $n \ge v_G$ and 
$m \ge \frac{9e_G^2}{\epsilon}$.
Set $p := \frac{m}{n}$.
There exists $c_G > 0$ depending only on $G$
such that:
\begin{align*}
\Pr_{e \leftarrow \mathsf{G}_{n,m}} \Big[ q(e)
\ge \mu(1+\epsilon) \Big]
\le
\exp(-c_G \epsilon^2 M^*_G(n, p)) \; .
\end{align*}
\end{theorem}
\begin{proof}
Since $\mathsf{G}_{n,m}$ is also $(0, \binom{n}{2})$-almost independent,
the only issue is bounding $\frac{\mu^*}{\mu}$. Our constraints give:
\begin{align*}
\frac{\mu^*}{\mu} \le \Big(1 + \frac{e_G}{m-e_G}\Big)^{e_G}
\le \Big(1 + \frac{\epsilon}{8e_G}\Big)^{e_G}
\le \exp\Big(\frac{\epsilon}{8}\Big) \le 1 + \frac{\epsilon}{4} \; 
\end{align*}
($\exp(\epsilon) \le 1 + \epsilon/2$ for $\epsilon \in [0, \frac{1}{4}]$).
With this bound in mind we apply 
Theorem \ref{thm:graph-gb} and Corollary \ref{cor:gb}.1
setting $m := c'_G\epsilon M_G^*(n, p)$:
\begin{align*}
\Pr_{e\leftarrow \mathsf{G}_{n,m}} 
\Big[ q(e) \ge \mu(1+\epsilon)\Big]
& \le
\Pr_{e\leftarrow \mathsf{G}_{n,m}} 
\Big[ q(e) \ge \mu^*\Big(\frac{1+\epsilon}{1+\epsilon/4}\Big)\Big]
\\ & \le
\Pr_{e\leftarrow \mathsf{G}_{n,m}} 
\Big[ q(e) \ge \mu^*\big(1+\frac{\epsilon}{2}\big) \Big]
\\ & \le
\exp\big(-c_G \epsilon^2 M^*_G\big(n, p\big)\big) \; .
\end{align*}
\end{proof}

\newpage
\bibliography{stacs29Hazla}

\end{document}